\documentclass[letterpaper, 10 pt, journal, onecolumn]{ieeetran} 

\IEEEoverridecommandlockouts                              
\usepackage{xcolor}
\usepackage{cite}
\usepackage{amsmath,amssymb,amsfonts}
\usepackage{algorithmic}
\usepackage{algorithm}
\usepackage{graphicx}
\usepackage{subfig}
\usepackage{textcomp}

\usepackage{mathtools}
\usepackage{amsthm}

\newtheorem{assumption}{Assumption}
\newtheorem{definition}{Definition}
\newtheorem{remark}{Remark}
\newtheorem{lemma}{Lemma}
\newtheorem{theorem}{Theorem}

\newcommand{\norm}[1]{\left\lVert #1 \right\rVert}

\date{}

\title{A novel constraint tightening approach for robust data-driven predictive control$^\dagger$\thanks{$^\dagger$This work was supported by Deutsche Forschungsgemeinschaft (DFG, German Research Foundation) under grant MU 3929/1-2 and AL 316/12-2 - 279734922, under Germany's Excellence Strategy - EXC 2075 - 390740016, and under grant 468094890. We acknowledge the support by the Stuttgart Center for Simulation Science (SimTech). The authors thank the International Max Planck Research School for Intelligent Systems (IMPRS-IS) for supporting Julian Berberich.}}

\author{Christian~Kl{\"o}ppelt$^*$,~Julian~Berberich$^{**}$,~Frank~Allg{\"o}wer$^{**}$,~Matthias~A.~M{\"u}ller$^*$%
\thanks{$^*$Christian Kl{\"o}ppelt and Matthias A. M{\"u}ller are with the Institute of Automatic Control, Leibniz University Hannover, Germany.}
\thanks{$^{**}$Julian Berberich and Frank Allg{\"o}wer are with the Institute for Systems Theory and Automatic Control, University of Stuttgart, Germany}}

\begin{document}

\maketitle
\thispagestyle{plain}
\pagestyle{plain}

\begin{abstract}
	In this paper, we present a data-driven model predictive control (MPC) scheme that is capable of stabilizing unknown linear time-invariant systems under the influence of process disturbances. To this end, Willems' lemma is used to predict the future behavior of the system. This allows the entire scheme to be set up using only a priori measured data and knowledge of an upper bound on the system order. First, we develop a state-feedback MPC scheme, based on input-state data, which guarantees closed-loop practical exponential stability and recursive feasibility as well as closed-loop constraint satisfaction. The scheme is extended by a suitable constraint tightening, which can also be constructed using only data. In order to control a priori unstable systems, the presented scheme contains a pre-stabilizing controller and an associated input constraint tightening. We first present the proposed data-driven MPC scheme for the case of full state measurements, and also provide extensions for obtaining similar closed-loop guarantees in case of output feedback. The presented scheme is applied to a numerical example.
\end{abstract}


\section{Introduction}\label{sec:introduction}

In recent years, there has been significant interest in designing data-driven model predictive control (MPC) schemes, in which predictions are not based on a parametric model of the system, but rather directly on a priori collected input/output data from the system, thus circumventing the challenging intermediate step of finding an accurate model. This is done by employing the so-called Willems' fundamental lemma\cite{Willems2005}, which states that for a controllable linear system, all possible system trajectories can be parametrized in terms of linear combinations of time-shifts of one single, persistently exciting, trajectory.

A direct data-based MPC scheme based on Willems' lemma was first considered by Yang et al.\cite{Yang2015} and Coulson et al.\cite{Coulson2019a}. Guarantees for recursive feasibility, stability, and robustness (in the presence of measurement noise) of the closed loop were first proven by Berberich et al.\cite{Berberich2020a}. In recent years, various further properties and extensions of this data-driven MPC framework have been studied, compare, e.g., the works by Coulson et al.\cite{Coulson2021}, Huang et al.\cite{Huang2021a}, Yin et al.\cite{Yin2021a,Yin2021b}, Xue and Matni\cite{Xue2021}, Furieri et al.\cite{Furieri2021}, Berberich et al.\cite{Berberich2021} and the overview paper by Markovsky and Dörfler\cite{Markovsky2021}. 

One of the major strengths of MPC is its ability to take constraints in the optimization problem into account, and therefore, guarantee their satisfaction in closed-loop operation. For data-driven MPC schemes, achieving constraint satisfaction is similarly important. However, in practice we typically only have access to noisy data. Thus, to achieve closed-loop constraint satisfaction, a suitable constraint tightening is required, similar to model-based MPC\cite{Kouvaritakis2016,Chisci2001,Mayne2005}. Such a constraint tightening takes into account possible (worst-case) disturbances as well as their influence on the system dynamics in order to ensure that no disturbance that may occur in the future can result in constraint violation. 

For data-driven MPC schemes relying on an a priori identification of the system model, there already exist schemes that provide a proper constraint tightening even in the case of additive process noise\cite{Aswani2013,Terzi2019}. However, in the direct data-driven setting --based on Willems' lemma-- this problem has not been conclusively solved so far. In Berberich et al.\cite{Berberich2020b}, closed-loop constraint satisfaction is shown in case of measurement noise; however, no process noise and no input constraint tightening is considered, resulting in the fact that the proposed scheme can only be applied (without being overly conservative) to open-loop stable systems. Process disturbances acting additively on the dynamics have been considered by Huang et al. \cite{Huang2021} and Umenberger et al.\cite{Umenberger2021}. However, both schemes lack the aforementioned closed-loop guarantees, and moreover, rely on the knowledge of a priori measured disturbances. Recently, Liu et al.\cite{Liu2021} proposed a scheme guaranteeing closed-loop stability and recursive feasibility in the presence of process disturbances, which, however, lacks of guarantees for closed-loop constraint satisfaction.

In this paper, we propose a data-driven MPC scheme for which recursive feasibility, robust stability, and closed-loop constraint satisfaction can be ensured in the presence of process disturbances. In particular, we adapt the state/output constraint tightening originally proposed by Berberich et al.\cite{Berberich2020b} to the case of process disturbances and propose a suitable additional input constraint tightening to allow for a pre-stabilizing feedback, such that the scheme can also be applied to open-loop unstable systems. We first consider the case where full state measurement is available, before presenting extensions to the output feedback case. We discuss how certain constants required for the proposed constraint tightening can be computed from a priori collected data. For both schemes, guarantees for practical exponential stability, recursive feasibility and closed-loop constraint satisfaction are proven.

The remainder of the paper is structured as follows. In Section \ref{sec:setup} the problem setup and preliminaries, such as the concept of persistency of excitation and Willems' lemma, are introduced. Next, in Section \ref{sec:state_feedback}, we set up the first proposed data-driven MPC scheme based on state measurements. To this end, we elaborate the data-driven parametrization of the constraint tightening, explain the MPC scheme and prove the aforementioned closed-loop guarantees. Thereafter, in Section \ref{sec:output_feedback} we present data-driven output-feedback MPC based on the case where only output measurements are available. We apply the proposed scheme to a numerical example in Section \ref{sec:examples}, and end with some concluding remarks in Section \ref{sec:conclusion}.

{\emph{Notation:} For a sequence $\{z_k\}_{k=0}^{N-1}$, we define the Hankel matrix of depth $L$ as
\begin{equation*}
	H_L(z) = 
	\begin{bmatrix} 
		z_0 & z_1 & \dots & z_{N-L} \\
		z_1 & z_2 & \dots & z_{N-L+1} \\
		\vdots & \vdots & \ddots & \vdots \\
		z_{L-1} & z_L & \dots & z_{N-1}
	\end{bmatrix},
\end{equation*}
and the stacked window from time instant $a$ to $b$ as
\begin{equation*}
	z_{\left[a,b\right]} = \begin{bmatrix} z_a \\ \vdots \\ z_b \end{bmatrix}.
\end{equation*}}

\section{Problem setup and preliminaries}\label{sec:setup}

In this paper, we consider the discrete-time multi-input multi-output LTI system
\begin{equation}
	\begin{split}
		x_{k+1} & = Ax_k + Bu_k + w_k, \\
		y_k & = Cx_k + Du_k,
	\end{split}
	\label{eq:lti_sys}
\end{equation}
with the state $x_k\in\mathbb{R}^n$, the input $u_k\in\mathbb{R}^m$, the output $y_k\in\mathbb{R}^p$ and the process disturbance $w_k\in\mathbb{R}^n$. The setup can be extended to include measurement noise as well, see Remark \ref{re:rem1} below. Throughout the paper, we assume that \eqref{eq:lti_sys} is a minimal realization, i.e., that the pair $\left( A,B\right)$ is controllable and the pair $\left( A,C\right)$ is observable. Moreover, we consider the matrices $A,B,C,D$ as being unknown and the only knowledge about the system available being its order $n$.

Moreover, we assume that the process disturbances belongs to a hypercube (precise definitions will be given in Sections \ref{sec:state_feedback} and \ref{sec:output_feedback}). The goal of this paper is to construct a data-driven MPC scheme that stabilizes the origin and ensures input and state constraint satisfaction (cf. Section \ref{sec:state_feedback}) or output constraint satisfaction (cf. Section \ref{sec:output_feedback}), where the respective constraint sets are given by hypercubes.

To this end, we apply a persistently exciting (p.e.) input sequence to the system, and measure the resulting state/output sequence, where a persistently exciting sequence is defined as follows.
\begin{definition}
	A sequence $\{u_k\}_{k=0}^{N-1}$, with $u_k\in\mathbb{R}^m$, is persistently exciting of order $L$ if $\mathrm{rank}\left(H_L(u)\right)= mL$.
\end{definition}
We want to make use of Willems' fundamental lemma for the prediction in an MPC problem.
\begin{lemma}[Willems' lemma \cite{Willems2005}]\label{lem:willems}
	Suppose $\{u_k,\hat{y}_k\}_{k=0}^{N-1}$ is a trajectory of the controllable system
	\begin{equation}
		\begin{split}
			\hat{x}_{k+1} & = A\hat{x}_k + Bu_k, \\
			\hat{y}_k & = C\hat{x}_k + Du_k,
			\label{eq:undisturbed_lti_sys}
		\end{split}
	\end{equation}
	and $u$ is persistently exciting of order $L + n$. Then, $\{\bar{u}_k,\bar{y}_k\}_{k=0}^{L-1}$ is a trajectory of System \eqref{eq:undisturbed_lti_sys} if and only if there exists $\alpha\in\mathbb{R}^{N-L+1}$ such that
	\begin{equation}
		\begin{bmatrix}H_L(u) \\ H_L(\hat{y}) \end{bmatrix}\alpha = \begin{bmatrix} \bar{u}\\ \bar{y} \end{bmatrix}.
	\end{equation}
\end{lemma}
This lemma states that in the absence of disturbances, i.e., if $w_k=0$ for all $k\geq 0$, all trajectories of system \eqref{eq:lti_sys} can be parametrized by linear combinations of time shifts of a priori measured, sufficiently exciting input/output trajectories. In the following two sections, we set up MPC schemes that use these trajectories for the prediction of the systems behavior.

\section{Data-driven state-feedback predictive control}\label{sec:state_feedback}

In this section, we present a robust data-driven state-feedback MPC scheme with closed-loop guarantees on stability and constraint satisfaction in the presence of process noise. In Subsection \ref{ssec:scheme_sf}, we introduce the data-driven MPC scheme for the case of available state measurements. Thereafter, in Subsection \ref{ssec:guarantees_sf} we prove the closed-loop guarantees of the introduced control scheme. Finally, in Subsection \ref{ssec:system_constants} we show how the system constants, which are used to set up the constraint tightening of the MPC scheme, can be approximated purely from data.

\subsection{Proposed MPC scheme}\label{ssec:scheme_sf}

For the first data-driven predictive control scheme, we consider the availability of full state measurement, i.e., $C=I$, $D=0$ in \eqref{eq:lti_sys}. Moreover, we assume that the process disturbance belongs to the hypercube $w_t\in \mathbb{W} = \left\{ w\in\mathbb{R}^n \mid \norm{w}_\infty \leq w_\mathrm{max} \right\}$ for all $t\geq 0$, where $w_\mathrm{max}\geq 0$ is known, and the input and state constraint sets are given by the hypercubes $u_t\in \mathbb{U} = \left\{ u\in\mathbb{R}^m \mid \norm{u}_\infty \leq u_\mathrm{max} \right\}$ and $x_t\in \mathbb{X} = \left\{ x\in\mathbb{R}^n \mid \norm{x}_\infty \leq x_\mathrm{max} \right\}$ for some $u_\mathrm{max}>0$, $x_\mathrm{max}>0$. As will become clear later in this section, it is crucial for the construction of a proper constraint tightening that the prediction model is stable. If this is not the case a priori (i.e., $A$ is unstable), then this can be enforced via a pre-stabilizing input parametrization
\begin{equation}
	u_k = Kx_k + \nu_k,
	\label{eq:pre-stab}
\end{equation}
as it is common, for example, in tube-based MPC\cite{Kouvaritakis2016,Chisci2001}. The state feedback matrix $K$ is chosen such that all eigenvalues of $A_K = A+BK$ strictly lie inside the unit disc. Such a pre-stabilizing controller can be computed purely from data, e.g., following the approaches by Berberich et al.\cite{Berberich2020c} or van Waarde et al.\cite{VanWaarde2020b}. Throughout this paper, we assume that such a controller is known a priori. In case of a stable system, the following scheme can be applied with $K=0$.

To make use of Lemma \ref{lem:willems} for the prediction of open-loop state sequences, we consider the input $\nu_k$ of the pre-stabilized system
\begin{equation}
	x_{k+1} = A_Kx_k + B\nu_k + w_k.
	\label{eq:pre_stab_sf}
\end{equation}
We now apply a p.e. input sequence $\{\nu_k^d\}_{k=0}^{N-1}$ of length $N$ to System \eqref{eq:pre_stab_sf}, and measure the associated disturbed state sequence $\{x^d_k\}_{k=0}^{N}$, where the superscript "$d$" denotes a priori collected data. 
\begin{assumption}\label{as:pe}
	The input sequence $\{\nu_k^d\}_{k=0}^{N-1}$ is persistently exciting of order $L+n+1$.
\end{assumption}
We denote the cumulated disturbance influencing the collected data at time $k$ as
\begin{equation}
	d^d_k = \sum_{i=0}^{k-1} A^{k-1-i}_Kw_i^d
	\label{eq:data_disturbance}
\end{equation}
and the undisturbed state at as
\begin{equation}
	\hat x_k^d = x_k^d - d^d_k.
	\label{eq:undisturbed_data}
\end{equation}

\begin{remark} \label{re:rem1}
	Note that the following results can be easily extended to the case where, apart from the process disturbance $w_k$, also additive measurement noise on the state measurements is present. In this case, the disturbance sequence in \eqref{eq:data_disturbance} has to be extended by the actual measurement noise instant occurring at time $k$ and, in case of $K\neq 0$, the past measurement noise which are fed back into and propagated through the system dynamics. However, for the sake of simplicity, we consider only process disturbances throughout the paper.
\end{remark}

With these a priori generated data sequences, we are now able to set up the following optimal control problem (OCP), given the measured state $x_t$ at time $t$ and with the prediction horizon $L$
\begin{subequations}
	\label{eq:ocp_sf}
	\begin{alignat}{4}
		&J_L^\ast(x_{t}) = & \min_{\substack{\alpha(t), \sigma(t), \\ \bar{\nu}(t), \bar{x}(t) }} & \sum_{k=0}^{L-1} \left(\norm{\bar{\nu}_k(t)}_R^2 + \norm{\bar{x}_k(t)}_Q^2 \right) + \lambda_\alpha w_\mathrm{max}\norm{\alpha(t)}_2^2 + \frac{\lambda_\sigma}{w_\mathrm{max}} \norm{\sigma(t)}_2^2 \label{eq:ocp_cost} \\
		&& \text{s.t.} & \begin{bmatrix} \bar{\nu}(t)\\ \bar{x}(t) + \sigma(t) \end{bmatrix} = \begin{bmatrix}H_{L}(\nu^d) \\ H_{L+1}(x^d) \end{bmatrix}\alpha(t), \label{eq:ocp_dyn} \\
		&&& \bar{x}_{0}(t) = x_{t}, \label{eq:ocp_init} \\
		&&& \bar{x}_{L}(t) = 0, \label{eq:ocp_term} \\
		&&& \norm{\bar x_k(t)}_\infty + a_{u,k}\norm{\bar \nu(t)}_1 + a_{\alpha,k}\norm{\alpha(t)}_1 + a_{\sigma,k}\norm{\sigma_k(t)}_\infty + a_{c,k} \leq x_\mathrm{max}, \label{eq:ocp_sc} \\
		&&& \norm{\bar \nu_k(t)}_\infty + b_{u,k}\norm{\bar \nu(t)}_1 + b_{\alpha,k}\norm{\alpha(t)}_1 + b_{\sigma,k}\norm{\sigma_k(t)}_\infty + b_{c,k} + \norm{K\bar x_k (t)}_\infty \leq u_\mathrm{max}, \label{eq:ocp_ic} \\
		&&& \forall k = 0,\dots , L-1.
	\end{alignat}
\end{subequations}
We denote the optimal solution of \eqref{eq:ocp_sf} at time $t$ by $\bar \nu^\ast (t)$, $\bar x^\ast (t)$, $\alpha^\ast (t)$, $\sigma^\ast (t)$. In \eqref{eq:ocp_dyn}, we make use of Lemma \ref{lem:willems} for the prediction of future state sequences of the system. Note that, a Hankel matrix of depth $L+1$ is used in the second block row of \eqref{eq:ocp_dyn}, since the predicted state sequence contains $L+1$ elements (from $k=0$ to $k=L$), whereas the predicted input sequence only contains $L$ elements (from $k=0$ to k=$L-1$). Moreover, note that, as it is common in predictive control based on Willems' lemma, we make use of a slack variable $\sigma (t)$ (first introduced by Coulson et al.\cite{Coulson2019a}) that renders \eqref{eq:ocp_dyn} feasible, even in the presence of disturbances. The slack variable $\sigma$ as well as the variable $\alpha$ are regularized in \eqref{eq:ocp_cost}. This leads to smaller values of $\sigma$ and $\alpha$, improving the prediction accuracy and reducing the influence of disturbances in the Hankel matrices. For further discussion on these issues, see also Section IV.A by Berberich et al.\cite{Berberich2020a} and Section IV by Dörfler et al. \cite{Doerfler2022}. Problem \eqref{eq:ocp_sf} contains the tightened state and input constraints \eqref{eq:ocp_sc} and \eqref{eq:ocp_ic}. Note that both constraints depend on $\bar u(t)$, $\alpha (t)$ and $\sigma (t)$. Together with suitably defined coefficients $a_{u,k}$, $a_{\alpha,k}$, $a_{\sigma,k}$, $a_{c,k}$, and $b_{u,k}$, $b_{\alpha,k}$, $b_{\sigma,k}$, $b_{c,k}$, which will be defined later on, this constrained tightening ensures recursive feasibility and closed-loop constraint satisfaction (cf. Theorem \ref{th:sf}), i.e., $\norm{x_t}_\infty\leq x_\mathrm{max}$ and $\norm{u_t}_\infty = \norm{Kx_t + \nu_t}_\infty \leq u_\mathrm{max}$ for all $t\geq 0$. Constraint \eqref{eq:ocp_ic} can be dropped if no pre-stabilizing controller is used, i.e., $K=0$. Problem \eqref{eq:ocp_sf} is similar to the one by Berberich et al.\cite{Berberich2020b}, which, however, does not consider process disturbances and the input constraint tightening \eqref{eq:ocp_ic}. The predictive control scheme is used in an $n$-step receding horizon manner, i.e., at time $t$ we solve \eqref{eq:ocp_sf} and choose $\nu_{t+k} = \bar \nu_k^\ast(t)$ in \eqref{eq:pre-stab} for $k=0,\dots,n-1$. 

In the following, we introduce the coefficients used to set up the tightened state and input constraints \eqref{eq:ocp_sc} and \eqref{eq:ocp_ic}. To this end, we first introduce some system constants. We denote the disturbance sequence $\{w_i\}_{i=0}^{k}$ propagated $k$ steps through the system dynamics by
\begin{equation}
	d_k \coloneqq  \sum_{i=0}^{k-1} A^{k-1-i}_Kw_i,
	\label{eq:prop_disturbance}
\end{equation}
and an upper bound on its $\infty$-norm as
\begin{equation}
	\bar d_k \geq \sum_{i=0}^{k-1} \norm{A^{k-1-i}_K}_\infty w_\mathrm{max} \geq \norm{d_k}_\infty.
	\label{eq:disturbance_bound}
\end{equation}
Moreover, we define the constant $c_{pe} = \norm{H_{u\hat x}^\dagger}_1$, with
\begin{equation}
	H_{u\hat x} = \begin{bmatrix} H_L\left(\nu^d\right) \\ H_1\left(\hat x^d_{[0,N-L-1]}\right)\end{bmatrix},
\end{equation}
where $H_{u\hat x}^\dagger$ is the Moore-Penrose inverse of $H_{u\hat x}$. Using these system constants, as well as $\rho_{A,k} \geq \norm{A_K^k}_\infty$, we define
\begin{equation}
	c_{\alpha,k} = \rho_{A,k}\bar d_{N-L} + \bar d_{N-L+k}, \quad  c_{\sigma,k} = \rho_{A,k} + 1,
	\label{eq:const_perror}
\end{equation}
for $k=0,\dots,L$. Moreover, we introduce the controllability constant $\Gamma>0$, which is chosen such that, starting at any $x_0$, we can find an input sequence $\nu_{[0,n-1]}$ steering the state of the pre-stabilized system \eqref{eq:pre_stab_sf} to the origin in $n$ steps and satisfying
\begin{equation}
	\norm{\nu_{[0,n-1]}}_1 \leq \Gamma \norm{x_0}_\infty.
	\label{eq:controllability}
\end{equation}
Note that such a constant exists as the pair $\left(A,B\right)$ is controllable.

We are now ready to define the coefficients of the state and input constraint tightening as
\begin{align}
	\begin{split}
		a_{u,k} & = 0,\ a_{\alpha,k} = c_{\alpha,k},\ a_{\sigma,k}  = c_{\sigma,k},\ a_{c,k} = \bar d_k, \\
		b_{u,k} & = 0,\ b_{\alpha,k} = \bar K c_{\alpha,k},\ b_{\sigma,k}  = \bar K c_{\sigma,k},\ b_{c,k}  = \bar K \bar d_k,
	\end{split}
	\label{eq:tight_constants_1}
\end{align}
for $k = 0,\dots,n-1$, and
\begin{align}
	\begin{split}
		a_{u,k+n} & = a_{u,k} + a_{\alpha,k}c_{pe} + a_{\sigma,k}c_{pe}\bar d_{N-1}, \\
		a_{\alpha,k+n} & =  a_{u,k+n}\Gamma c_{\alpha,L-1} + c_{\alpha,k+n}, \\
		a_{\sigma,k+n} & = a_{u,k+n}\Gamma c_{\sigma,L-1} + c_{\sigma,k+n}, \\
		a_{c,k+n} & = a_{c,k} + a_{\alpha,k}c_{pe}\left(nx_\mathrm{max} + n\bar d_n\right) + a_{\sigma,k}\left(\bar d_{N-1}c_{pe}\left(nx_\mathrm{max} + n\bar d_n\right) + \bar d_n\right) + \bar d_n, \\
		b_{u,k+n} & = b_{u,k} + b_{\alpha,k}c_{pe} + b_{\sigma,k}c_{pe}\bar d_{N-1}, \\
		b_{\alpha,k+n} & = b_{u,k+n}\Gamma c_{\alpha,L-1} + \bar K c_{\alpha,k+n}, \\
		b_{\sigma,k+n} & = b_{u,k+n}\Gamma c_{\sigma,L-1} + \bar K c_{\sigma,k+n}, \\
		b_{c,k+n} & = b_{c,k} + b_{\alpha,k}c_{pe}\left(nx_\mathrm{max} + n\bar d_n\right) + b_{\sigma,k}\left(\bar d_{N-1}c_{pe}\left(nx_\mathrm{max} + n\bar d_n\right) + \bar d_n\right) + \bar K\bar d_n,
	\end{split}
	\label{eq:tight_constants_2}
\end{align}
for $k = 0,\dots,L-n-1$, where $\bar K = \norm{K}_\infty$. Note that $\bar d_k$ and $\rho_{A,k}$ grow exponentially if $A_K$ has eigenvalues outside the unit disc. This is the main motivation for the usage of the pre-stabilizing controller \eqref{eq:pre-stab}, as diverging $\bar d_k$ and $\rho_{A,k}$ would also lead to diverging $c_{\alpha,k}$ and $c_{\sigma,k}$ and, therefore, to large coefficients \eqref{eq:tight_constants_1} and \eqref{eq:tight_constants_2}. This would in general yield an infeasible OCP \eqref{eq:ocp_sf} even for small prediction horizons $L$. In order to set up the coefficients above, the system constants $\Gamma$, $c_{pe}$, $\rho_{A,k}$ for $k=0,\dots,L$, and $d_k$ for $k=0,\dots,N-1$ have to be known. All of these constants can be approximated from data as will be shown in Subsection \ref{ssec:system_constants}.

\subsection{Theoretical guarantees}\label{ssec:guarantees_sf}

Firstly, we denote the undisturbed state at time $t+k$ resulting from an open-loop application of $\bar \nu^\ast(t)$ as
\begin{equation}
	\hat x_{t+k}^\ast \coloneqq A_K^k x_t + \sum_{i=0}^{k-1}A_K^{k-1-i}B\bar \nu_i^\ast(t).
	\label{eq:ol_nom_state}
\end{equation}
An upper bound for the prediction error between this undisturbed open-loop state trajectory $\hat x^\ast$ and the predicted optimal state sequence $\bar x^\ast(t)$ at time $t$ can be derived by the following lemma.
\begin{lemma}
	If \eqref{eq:ocp_sf} is feasible at time t, then
	\begin{equation}
		\label{eq:prediction_error_sf}
		\norm{\hat{x}_{t+k}^\ast - \bar x^\ast_k(t)}_\infty \leq c_{\alpha,k}\norm{\alpha^\ast(t)}_1 + c_{\sigma,k}\norm{\sigma^\ast(t)}_\infty
	\end{equation}
	holds for all $k = 0,\dots,L$.
	\label{lem:pred_error_sf}
\end{lemma}
\begin{proof}
	Similar to the proof of Lemma 2 in the work of Berberich et al.\cite{Berberich2020a}, we start by considering the error between the undisturbed open-loop state and the state prediction resulting from undisturbed data in the Hankel matrix
	\begin{equation}
		x^-_{[t,t+L]} \coloneqq \hat x_{[t,t+L]}^\ast - H_{L+1}\left( \hat x^d\right)\alpha^\ast(t) = \hat x_{[t,t+L]}^\ast - \left(H_{L+1}\left( x^d\right)- H_{L+1}\left( d^d\right)\right)\alpha^\ast(t).
	\end{equation}
	Note that $\hat x_{[t,t+L-1]}^\ast$ and $H_{L}\left( \hat x^d\right)\alpha^\ast(t)$ are trajectories of the undisturbed system \eqref{eq:pre_stab_sf} resulting from the application of $\bar \nu^\ast(t)$ with different initial states. Due to \eqref{eq:ocp_init} and \eqref{eq:ocp_dyn}, the initial condition of $H_{L+1}\left(\hat x^d\right)\alpha^\ast (t)$ is
	\begin{equation}
		H_1\left(\hat x^d_{[0,N-L]}\right)\alpha^\ast (t) = x_t + \sigma_0 (t) - H_1\left(\hat d^d_{[0,N-L]}\right)\alpha^\ast (t).
	\end{equation}
	Thus, for the difference between both trajectories, it holds that
	\begin{equation}
		x^-_{t+k} = A_K^k\left(H_1\left( d^d_{[0,N-L]}\right)\alpha^\ast (t) - \sigma_0 (t) \right)
	\end{equation}
	for $k=0,\dots,L$. To show \eqref{eq:prediction_error_sf}, note that
	\begin{align}
		\norm{\hat{x}_{t+k}^\ast - \bar x^\ast_k(t)}_\infty & \leq \norm{x^-_{t+k}}_\infty + \norm{\sigma_k^\ast(t)- H_1\left( d^d_{[k,N-L+k]}\right)\alpha^\ast (t)}_\infty \\
			& \leq \left(\norm{A_K^k}_\infty\bar d_{N-L} + \bar d_{N-L-1}\right)\norm{\alpha^\ast (t)}_1 + \norm{A_K^k}_\infty\norm{\sigma_0^\ast(t)}_\infty + \norm{\sigma_k^\ast(t)}_\infty,
	\end{align}
	where the second inequality holds due to
	\begin{align}
		\norm{H_1\left( d^d_{[k,N-L+k]}\right)\alpha^\ast (t)}_\infty & \leq \norm{d^d_{[k,N-L+k]}}_\infty\norm{\alpha^\ast (t)}_1\\
		& \leq \bar d_{N-L+k}\norm{\alpha^\ast (t)}_1
	\end{align}
	for $k=0,\dots,L$. Therefore, we obtain \eqref{eq:prediction_error_sf} with $c_{\alpha,k} = \rho_{A,k}\bar d_{N-L} + \bar d_{N-L+k}$, $c_{\sigma,k} = \rho_{A,k} + 1$, and $\rho_{A,k} \geq \norm{A_K^k}_\infty$.	
\end{proof}

Using the result of Lemma \ref{lem:pred_error_sf}, we can now state our main result, which establishes recursive feasibility, practical exponential stability, and input and state constraint satisfaction of the closed-loop system, assuming that the initial state is feasible for Problem \eqref{eq:ocp_sf} and the disturbance bound is sufficiently small.

\begin{theorem} \label{th:sf}
	Suppose that Assumption \ref{as:pe} holds. Then, for any $V_\mathrm{ROA} > 0$, there exist $\underline{\lambda}_\alpha$, $\overline{\lambda}_\alpha$, $\underline{\lambda}_\sigma$, $\overline{\lambda}_\sigma$ such that for all $\lambda_\alpha$, $\lambda_\sigma$ satisfying
	\begin{equation}
		\underline{\lambda}_\alpha \leq \lambda_\alpha \leq \overline{\lambda}_\alpha, \quad \underline{\lambda}_\sigma \leq \lambda_\sigma \leq \overline{\lambda}_\sigma,
	\end{equation}
	there exist $\bar w_\mathrm{max}$, $\bar c_{pe} >0$ as well as a continuous, strictly increasing function $\beta : [0,\bar w_\mathrm{max}] \to [0, V_\mathrm{ROA}]$ with $\beta(0) = 0$, such that for all $w_\mathrm{max}$ and $c_{pe}$ satisfying
	\begin{equation}
		w_\mathrm{max} \leq \min{\left\{\bar w_\mathrm{max},\; \frac{\bar c_{pe}}{c_{pe}}\right\}}
	\end{equation}
	the following holds for the closed loop resulting from an application of the $n$-step MPC scheme:
	\begin{enumerate}
		\item[(i)] \label{thrm_sf_i} If $J_L^\ast(x_0) \leq V_\mathrm{ROA}$, then OCP \eqref{eq:ocp_sf} is feasible at any time $t\geq 0$.
		\item[(ii)] \label{thrm_sf_ii} For any initial condition satisfying $J_T^\ast(x_0) \leq V_\mathrm{ROA}$ it holds that $x_t\in\mathbb{X}$ and $u_t\in\mathbb{U}$ for all $t\geq 0$, and $J_L^\ast(x_t)$ converges exponentially to $J_L^\ast(x_t) \leq \beta (\bar w_\mathrm{max})$.
	\end{enumerate}
\end{theorem}
The proof of this result is similar to the one of Theorem 10 by Berberich et al.\cite{Berberich2020b}. The main differences lie in the fact that we consider process instead of measurement noise; moreover, we use a pre-stabilizing controller, and thus, include the input constraint tightening \eqref{eq:ocp_ic}, which also has to be satisfied by the considered candidate solution, into the OCP. Furthermore, note that (ii) only shows exponential convergence of $x_t$ to a neighborhood of $x=0$, however, it is possible to establish a suitable lower as well as an upper bound on $J_L^\ast(x_t)$ analogous to Lemma 1 by Berberich et al.\cite{Berberich2020a}, thus, resulting in practical exponential stability. For a detailed discussion on the influence of the parameters $\lambda_\alpha$, $\lambda_\sigma$, and $c_{pe}$ on the stability properties, we refer to the work of Berberich et al.\cite{Berberich2020a}. In short, the region of attraction increases for smaller disturbance bounds or better persistency of excitation of the input signal, the latter being expressed by a decrease of $c_{pe}$.
\begin{proof}
	To show (i), we construct a candidate solution for \eqref{eq:ocp_sf} at time $t+n$ and show that \eqref{eq:ocp_dyn}-\eqref{eq:ocp_ic} hold for this candidate. Therefore, we define the candidate solution over the first $L-n$ steps via the previously optimal input shifted by $n$ steps, i.e., $\bar \nu'_k(t+n) = \bar \nu^\ast_{k+n}(t)$ for $k = 0,\dots,L-n-1$. Moreover, the state candidate is chosen as
	\begin{equation}
		\bar x'_{[0,L-n]}(t+n) = 
		\begin{bmatrix}
			x_{t+n} \\
			\hat x^\ast_{[t+n+1,t+L]}
		\end{bmatrix}
	\end{equation}
	for the first $L-n+1$ steps, with $\hat x^\ast$ defined as in \eqref{eq:ol_nom_state}. Note that, $c_{\alpha,k}$ scales linearly with $w_\mathrm{max}$ and 
	\begin{equation}
		\norm{\alpha^\ast(t)}_1 \leq \sqrt{N-L+1} \norm{\alpha^\ast(t)}_2 \leq \sqrt{\left(N-L+1\right) \frac{V_\mathrm{ROA}}{\lambda_\alpha w_\mathrm{max}}},
	\end{equation}
	which implies that the first term on the right hand side of \eqref{eq:prediction_error_sf} becomes arbitrarily small for sufficiently small $w_\mathrm{max}$. Furthermore, the same holds true for the second term on the right hand side of \eqref{eq:prediction_error_sf}, since $c_{\sigma,k}$ is uniformly upper bounded for all $k$ and
	\begin{equation}
		\norm{\sigma^\ast (t)}_\infty \leq \norm{\sigma (t)}_2 \leq \sqrt{\frac{V_\mathrm{ROA}w_\mathrm{max}}{\lambda_\sigma}}.
	\end{equation}
	Due to this, Lemma \ref{lem:pred_error_sf}, and $\bar x^\ast_{L}(t) = 0$, the state $\bar x'_{L-n}(t+n) = \hat x^\ast_{t+L}$ becomes arbitrarily small for sufficiently small $w_\mathrm{max}$. Thus, for a sufficiently small $w_\mathrm{max}$, by controllability, there exists an input sequence $\bar \nu'_{[L-n,L-1]}(t+n)$ that steers $\bar x'_{[L-n,L]}(t+n)$ from $\hat x^\ast_{t+L}$ to $0$ in $n$ steps. Furthermore, we choose the candidate solution for $\alpha'(t+n)$ and $\sigma'(t+n)$ as
	\begin{equation}
		\alpha'(t+n) = H_{u\hat x}^\dagger \begin{bmatrix} \bar \nu'(t+n) \\ \hat x_{t+n}^\ast \end{bmatrix},
		\label{eq:alpha_candidate}
	\end{equation}
	and
	\begin{equation}
		\begin{split}
		\sigma'(t+n) & = H_{L+1}\left(x^d\right)\alpha'(t+n) - \bar x'(t+n), \\
		& = \begin{bmatrix} \hat d_{t+n} \\ H_{L}\left(d_{[1,N]}^d\right)\alpha'(t+n)\end{bmatrix},
		\end{split}
		\label{eq:sigma_candidate}
	\end{equation}	
	where $\hat d_{t+k}$ satisfies  $x_{t+k} = \hat x_{t+k}^\ast + \hat d_{t+k}$. Thus, the candidate solution satisfies \eqref{eq:ocp_dyn}-\eqref{eq:ocp_term}.
	
	It remains to be shown that also \eqref{eq:ocp_sc} and \eqref{eq:ocp_ic} hold for the candidate solution. In order to show the satisfaction of \eqref{eq:ocp_ic}, we note that
	\begin{equation}
		\begin{split}
			\norm{\alpha'(t+n)}_1 & \stackrel{\eqref{eq:alpha_candidate}}{\leq} c_{pe}\left(\norm{\bar \nu'(t+n)}_1 + \norm{\hat x_{t+n}^\ast}_1 \right)\\
			& \leq c_{pe}\left(\norm{\bar \nu'(t+n)}_1 + nx_\mathrm{max} + n\bar d_n\right)
		\end{split}
		\label{eq:alpha_bound}
	\end{equation}
	holds. Furthermore, it holds that
	\begin{equation}
		\begin{split}
		\norm{\sigma' (t+n)}_\infty & \stackrel{\eqref{eq:sigma_candidate}}{\leq} \bar d_{N-1}\norm{\alpha'(t+n)}_1 + \bar d_n, \\
		& \stackrel{\eqref{eq:alpha_bound}}{\leq} \bar d_{N-1}c_{pe}\left(\norm{\bar \nu'(t+n)}_1 + nx_\mathrm{max} + n\bar d_n\right) + \bar d_n.
		\end{split}
		\label{eq:sigma_bound}
	\end{equation}
	Using the same arguments as in Proposition 8 (Inequality (18)) by Berberich et al.\cite{Berberich2020b}, we can bound the norm of the candidate input by
	\begin{equation}
		\norm{\bar \nu'(t+n)}_1 \leq \norm{\bar \nu^\ast(t)}_1 + \Gamma c_{\alpha, L-1}\norm{\alpha^\ast(t)}_1 + \Gamma c_{\sigma, L-1}\norm{\sigma^\ast(t)}_\infty.
		\label{eq:input_bound}
	\end{equation}
	Since \eqref{eq:ocp_ic} holds for the optimal solution at time $t$ and due to $\norm{\bar \nu_k'(t+n)}_\infty = \norm{\bar \nu^\ast_{k+n}}_\infty$ for $k = 0,\dots,L-n-1$, we obtain
	\begin{equation}
		\norm{\bar \nu_k'(t+n)}_\infty \leq u_\mathrm{max} - \left(b_{u,k+n}\norm{\bar \nu^\ast(t)}_1 + b_{\alpha,k+n}\norm{\alpha^\ast(t)}_1 + b_{\sigma,k+n}\norm{\sigma^\ast(t)}_\infty + b_{c,k+n} + \norm{K\bar x^\ast_{k+n}(t)}_\infty\right).
		\label{eq:rf_1}
	\end{equation}
	Plugging \eqref{eq:input_bound} and the coefficients in \eqref{eq:tight_constants_2} into \eqref{eq:rf_1} yields for $k=0,\dots,L-n-1$
	\begin{align}
		\begin{split}
			\norm{\bar \nu_k'(t+n)}_\infty & \stackrel{\eqref{eq:input_bound}}{\leq} u_\mathrm{max} - b_{u,k}\norm{\bar \nu'(t+n)}_1 - b_{\alpha,k}\left(c_{pe}\norm{\bar \nu'(t+n)}_1 + c_{pe}(nx_\mathrm{max} + n\bar d_n)\right) \\
			& \quad - b_{\sigma,k}\left(\bar d_{N-1}c_{pe}\norm{\bar \nu'(t+n)}_1 + \bar d_{N-1}c_{pe}(nx_\mathrm{max} + n\bar d_n) + \bar d_n \right) - b_{c,k}\\
			& \quad - \bar K c_{\alpha, k+n}\norm{\alpha^\ast(t)}_1 - \bar K c_{\sigma, k+n}\norm{\sigma^\ast(t)}_\infty - \norm{K\bar x^\ast_{k+n}(t)}_\infty - \bar K\bar d_n,
		\end{split} \\
		\begin{split}
			& \stackrel{\eqref{eq:alpha_bound},\eqref{eq:sigma_bound}}{\leq} u_\mathrm{max} - b_{u,k}\norm{\bar \nu'(t+n)}_1 - b_{\alpha,k}\norm{\alpha'(t+n)}_1 - b_{\sigma,k}\norm{\sigma'(t+n)}_\infty - b_{c,k} \\
			& \quad - \bar K c_{\alpha, k+n}\norm{\alpha^\ast(t)}_1 - \bar K c_{\sigma, k+n}\norm{\sigma^\ast(t)}_\infty  - \norm{K\bar x^\ast_{k+n}(t)}_\infty - \bar K\bar d_n, 
		\end{split} \\
		\begin{split}
			& \stackrel{\eqref{eq:prediction_error_sf}}{\leq} u_\mathrm{max} - b_{u,k}\norm{\bar \nu'(t+n)}_1 - b_{\alpha,k}\norm{\alpha'(t+n)}_1 - b_{\sigma,k}\norm{\sigma'(t+n)}_\infty - b_{c,k} \\
			& \quad - \bar K \norm{\hat x_{t+k+n}^\ast - \bar x^\ast_{k+n}(t)}_\infty - \norm{K\bar x^\ast_{k+n}(t)}_\infty - \bar K\bar d_n,
		\end{split} \\
		\begin{split}
			& \leq u_\mathrm{max} - b_{u,k}\norm{\bar \nu'(t+n)}_1 - b_{\alpha,k}\norm{\alpha'(t+n)}_1 - b_{\sigma,k}\norm{\sigma'(t+n)}_\infty - b_{c,k} - \norm{K\bar x'_k(t+n)}_\infty,
		\end{split}
	\end{align}
	where the last inequality holds due to $\bar x_k'(t+n) = \hat x_{t+n+k}^\ast$ for $k = 1,\dots,L-n-1$, and
	\begin{equation}
		\norm{\bar x'_0(t+n) - \hat x_{t+n}^\ast}_\infty = \norm{x_{t+n} - \hat x_{t+n}^\ast}_\infty \leq \bar d_n. 
	\end{equation}
	Therefore, the candidate solution satisfies \eqref{eq:ocp_ic} for $k=0,\dots, L-n-1$.	Showing that also \eqref{eq:ocp_sc} holds for $k=0,\dots, L-n-1$ can be done following the analogous steps as above.
	
	To show that \eqref{eq:ocp_sc} and \eqref{eq:ocp_ic} are also satisfied for $k=L-n,\dots,L-1$, we recall from above that $\bar x'_{L-n}(t+n)$ becomes arbitrarily small for sufficiently small $w_\mathrm{max}$. Thus, due to controllability and \eqref{eq:controllability} also $\bar \nu_{[L-n,L-1]} '(t+n)$ and, therefore, $\bar x_{[L-n+1,L-1]}'(t+n)$ become arbitrarily small. Moreover, due to \eqref{eq:const_perror} and \eqref{eq:tight_constants_2} the coefficients $a_{\alpha,k}$ and $b_{\alpha,k}$ depend linearly on $w_\mathrm{max}$, and thus, they also become arbitrarily small for sufficiently small disturbance bounds $w_\mathrm{max}$. Moreover, the same holds for $a_{u,k}$, $a_{c,k}$, $b_{u,k}$, and $b_{c,k}$. The coefficients $a_{\sigma,k}$, $b_{\sigma,k}$ converge to constant values for $w_\mathrm{max} \to 0$.Finally, due to \eqref{eq:sigma_bound}, $\sigma'(t+n)$ becomes arbitrarily small if $w_\mathrm{max}$ is sufficiently small. Hence, all terms on the left-hand side of \eqref{eq:ocp_sc} and \eqref{eq:ocp_ic} (except for $\bar{x}_k(t)$ and $\bar{v}_k(t)$, respectively) become arbitrary small. Therefore, since $x_\mathrm{max},\ u_\mathrm{max} > 0$, there exists a sufficiently small bound $\bar w_\mathrm{max} > 0$ such that \eqref{eq:ocp_sc} and \eqref{eq:ocp_ic} are also satisfied for $k=L-n,\dots,L-1$. 
	
	To show (ii), we can follow the same arguments as in Theorem 3 by Berberich et al.\cite{Berberich2020a} to conclude invariance of the sublevel set $J_L^\ast(x_t) \leq V_\mathrm{ROA}$ and exponential convergence of $J_L^\ast(x_t)$ to $J_L^\ast(x_t) \leq \beta (\bar w_\mathrm{max})$. This is possible since the candidate solution used in the first part of our proof is constructed analogous to the one used in the proof of the above reference. Thus, the closed-loop scheme is recursively feasible. Closed-loop state constraint satisfaction follows from the same arguments used in the proof of Theorem 10 by Berberich et al.\cite{Berberich2020b}. In order to show closed-loop input constraint satisfaction, we note that the optimal solution at time $t$ satisfies \eqref{eq:ocp_ic} for the first $n$ steps, i.e.,
	\begin{align}
		u_\mathrm{max} & \geq \norm{\bar \nu^\ast_k(t)}_\infty + b_{u,k}\norm{\bar \nu^\ast(t)}_1 + b_{\alpha,k}\norm{\alpha^\ast(t)}_1 + b_{\sigma,k}\norm{\sigma^\ast(t)}_\infty + b_{c,k} + \norm{K \bar x^\ast_k(t)}_\infty \\
		& = \norm{\bar \nu^\ast_k(t)}_\infty + \bar K c_{\alpha,k}\norm{\alpha^\ast(t)}_1 + \bar K c_{\sigma,k} \norm{\sigma^\ast(t)}_\infty + \bar K\bar d_k + \norm{K\bar x_k^\ast(t)}_\infty,
	\end{align}
	for $k=0,\dots,n-1$, where the second inequality follows from plugging in the coefficients from \eqref{eq:tight_constants_1}. With \eqref{eq:prediction_error_sf}, we finally obtain for $k=0,\dots,n-1$
	\begin{align*}
		u_\mathrm{max} & \geq \norm{\bar \nu^\ast_k(t)}_\infty + \bar K \norm{\hat x_{t+k}^\ast - \bar x^\ast_k(t)}_\infty + \bar K\bar d_k + \norm{K\bar x_k^\ast(t)}_\infty \\
		& \geq \norm{\bar \nu^\ast_k(t)}_\infty + \bar K \norm{\hat x_{t+k}^\ast - \bar x^\ast_k(t)}_\infty + \bar K\norm{x_{t+k} - \hat x_{t+k}^\ast}_\infty + \norm{K\bar x_k^\ast(t)}_\infty \\
		& \geq \norm{\bar \nu^\ast_k(t) + Kx_{t+k}}_\infty,
	\end{align*}
	where the second inequality holds due to $\norm{x_{t+k} - \hat x^\ast_{t+k}}_\infty \leq \bar d_k$ and the third inequality due to the triangle inequality. Thus, the input constraints are satisfied in closed loop.
\end{proof}

\subsection{Data-driven Estimation of System Constants}\label{ssec:system_constants}

In order to set up the constraint tightening, we need to compute the coefficients \eqref{eq:tight_constants_1} and \eqref{eq:tight_constants_2}. These depend on several system constants. While $\bar K$, $x_\mathrm{max}$, and $n$ are known a priori, the system constants $\rho_{A,k}$, $\bar d_k$, $c_{pe}$ and $\Gamma$ have to be estimated from data. In the following, we provide corresponding estimation procedures.

An approach for the computation of $\Gamma$ was shown by Berberich et al.\cite{Berberich2020b}, where, however, the availability of the undisturbed data is assumed. To extend this approach for the approximation of $\Gamma$ in the presence of disturbances in the data, we adapt the respective optimization problem by including a slack variable as well as a regularization, similar to \eqref{eq:ocp_sf}. Thus, we set up the optimization problem
\begin{subequations}
\label{eq:gamma}
	\begin{alignat}{4}
	&&\max_{x_0}\min_{\substack{\alpha, \sigma, \\ \bar{\nu}, \bar{x} }}&\quad \norm{\bar\nu_{[0,n-1]}}_1 + \lambda_\alpha'w_\mathrm{max}\norm{\alpha}_2^2 + \frac{\lambda_\sigma'}{w_\mathrm{max}}\norm{\sigma}_2^2, \\
	&&\text{s.t.}& \quad \norm{\bar x_0}_\infty \leq x_\mathrm{max}, \\
	&&& \quad\bar x_n = 0, \\
	&&& \quad\begin{bmatrix} \bar{\nu}\\ \bar{x} + \sigma \end{bmatrix} = \begin{bmatrix}H_{n}(\nu^d) \\ H_{n+1}(x^d) \end{bmatrix}\alpha,
	\end{alignat}
\end{subequations}
which can be solved by solving the inner minimization problem for all vertices of $\mathbb{X}$. Then, we choose $\Gamma \approx \frac{\norm{\bar\nu^\ast_{[0,n-1]}}_1}{x_\mathrm{max}}$, where $\bar\nu^\ast_{[0,n-1]}$ is the optimal solution of \eqref{eq:gamma}. Moreover, we approximate $c_{pe} \approx  \norm{H_{ux}^\dagger}_1$, with
\begin{equation}
	H_{ux} = \begin{bmatrix} H_L\left(\nu^d\right) \\ H_1\left(x^d_{[0,N-L-1]}\right)\end{bmatrix}.
\end{equation}
This approximation is possible, as for small disturbances also the error between $H_{ux}$ and $H_{u\hat x}$ is small. So far, both procedures mentioned above only yield approximations of the real constants $\Gamma$ and $c_{pe}$, without being guaranteed overapproximations of these constants. However, as was also confirmed in our numerical examples, the error between the real values and our estimates remains small for the considered disturbance level. Obtaining guaranteed overapproximations of the constants $\Gamma$ and $c_{pe}$ even in the presence of noise is an interesting subject for future research.

In order to estimate the overapproximations $\rho_{A,k} \geq \norm{A_K^k}_\infty$ and $\bar d_k \geq \sum_{i=0}^{k-1} \norm{A_K^{k-1-i}}_\infty w_\mathrm{max}$ for $i=0,\dots, L-1$ we use Algorithm 1 by Wildhagen et al.\cite{Wildhagen2021}, which makes use of the S-lemma\cite{VanWaarde2020b}. Note that the setup in the aforementioned reference considers a bound on the $2$-norm of the disturbance. However, by noting that $\norm{w}_2 \leq \sqrt{n} \norm{w}_\infty$, we can easily adapt the algorithm to our setting.

\section{Data-driven output-feedback predictive control}\label{sec:output_feedback}

In this section, we construct a robust data-driven predictive control scheme, in case only output measurements are available. First, in Subsection \ref{ssec:scheme_of} we set up the MPC scheme and prove similar theoretical properties to the ones shown for the state-feedback case. Thereafter, in Subsection \ref{ssec:system_constants_of} we show how the coefficients used for the constraint tightening can be computed only from data. 

\subsection{Proposed MPC scheme and theoretical guarantees}\label{ssec:scheme_of}

In contrast to the previous section, we consider the case where no full state measurements are available, but only output measurements. To this end, we consider the difference operator form
\begin{equation}
	y_k = -A_n y_{k-1} - \ldots - A_2 y_{k-n+1} - A_1 y_{k-n} + D u_k + B_n u_{k-1} + \ldots + B_2 u_{k-n+1} + B_1 u_{k-n} + \tilde{w}_k,
	\label{eq:io_charact}
\end{equation}
with the process disturbance $\tilde{w}_k\in \mathbb{\tilde W} = \left\{ \tilde w\in\mathbb{R}^p \mid \norm{\tilde w}_\infty \leq \tilde w_\mathrm{max} \right\}$. The input and output constraints are given by $u_t\in \mathbb{U} = \left\{ u\in\mathbb{R}^m \mid \norm{u}_\infty \leq u_\mathrm{max} \right\}$ and  $y_t\in \mathbb{Y} = \left\{ y\in\mathbb{R}^p \mid \norm{y}_\infty \leq y_\mathrm{max} \right\}$ for some $u_\mathrm{max}>0$, $y_\mathrm{max}>0$, similar to the setup in Section \ref{sec:state_feedback}. Note that the following results also hold if only an upper bound on the system order is known, in which case $n$ needs to be replaced by this upper bound. Furthermore, note that \eqref{eq:io_charact} is an equivalent characterization of the input-output behavior of \eqref{eq:lti_sys}, with $C\neq I$ or $D\neq 0$. Moreover, we can transform \eqref{eq:io_charact} into the non-minimal realization
\begin{equation}
	\begin{split}
		\xi_{k+1} & = \tilde{A} \xi_k + \tilde{B}u_k + \tilde{E}\tilde{w}_k, \\
		\tilde y_k & = \tilde{C}\xi_k + D u_k + \tilde{w}_k,
	\end{split}
	\label{eq:extended_dyn}
\end{equation}
with the extended state $\xi_k = \begin{bmatrix} u_{k-n}^\top & \dots & u_{k-1}^\top & y_{k-n}^\top & \dots & y_{k-1}^\top\end{bmatrix}^\top$, cf. \cite{Goodwin2014}. Similar to the state feedback case in Section \ref{sec:state_feedback}, we want to make use of a pre-stabilizing controller in case of an unstable system. Therefore, we introduce the control law
\begin{equation}
	u_k = \tilde K\xi_k + \nu_k,
	\label{eq:pre-stab_of}
\end{equation}
where the stabilizing feedback matrix $\tilde K$ can be computed purely from data, e.g., following the approach by Berberich et al.\cite{Berberich2020d}. However, for simplicity it is assumed, that such a pre-stabilizing controller is known a priori. Thus, $\nu_k$ is the input to the stabilized system
\begin{equation}
	\begin{split}
		\xi_{k+1} & = \tilde{A}_K \xi_k + \tilde{B}\nu_k + \tilde{E}\tilde{w}_k, \\
		y_k & = \tilde{C}_K\xi_k + D \nu_k + \tilde{w}_k.
	\end{split}
	\label{eq:stab_sys_of}
\end{equation}
Similar to the previous section, we define the disturbance propagated $k$ steps through the system dynamics as
\begin{equation}
	\delta_k \coloneqq  \sum_{i=0}^{k-1} \tilde A^{k-1-i}_K \tilde w_i,
	\label{eq:prop_disturbance_of}
\end{equation}
and an upper bound on its $\infty$-norm as
\begin{equation}
	\bar \delta_k = \sum_{i=0}^{k-1} \eta_A^{k-1-i}\tilde w_\mathrm{max} \geq \sum_{i=0}^{k-1} \norm{\tilde A^{k-1-i}_K}_\infty \tilde w_\mathrm{max} \geq \norm{\delta_k}_\infty,
	\label{eq:disturbance_bound_of}
\end{equation}
with $\eta_A \geq \norm{\tilde A_K}_\infty$.

Again, we apply a p.e. input sequence $\{\nu_k^d\}_{k=0}^{N-1}$ of length $N$ to System \eqref{eq:stab_sys_of}, and measure the associated disturbed output sequence $\{y^d_k\}_{k=0}^{N-1}$. As the OCP, introduced in the following, now contains $n$ additional steps to fix the initial state, the following assumption is needed.
\begin{assumption}
	The input sequence $\{\nu_k^d\}_{k=0}^{N-1}$ is persistently exciting of order $L+2n$.
	\label{as:pe_of}
\end{assumption}
Using these a priori collected data sequences as well as the constants $\eta_A$ from above, $\eta_B \geq \norm{\tilde B}_\infty$, $\eta_C \geq \norm{\tilde C_K}_\infty$, $\eta_D \geq \norm{D}_\infty$, we set up the OCP for the output-feedback predictive control problem as
\begin{subequations} \label{eq:ocp_of}
	\begin{alignat}{4}
		&&& J_{L}^\ast(\bar \nu_{[0,n-1]}(t-n),y_{\left[t-n,t-1\right]}) =\\
		&& \min_{\substack{\alpha(t), \sigma(t), \\ \bar{\nu}(t), \bar{y}(t) }} & \sum_{k=0}^{L-1} \left(\norm{\bar{\nu}_k(t)}_R^2 + \norm{\bar{y}_k(t)}_Q^2 \right) + \lambda_\alpha \tilde{w}_\mathrm{max} \norm{\alpha(t)}_2^2 + \frac{\lambda_\sigma}{\tilde{w}_\mathrm{max}} \norm{\sigma(t)}_2^2 \\
		&& \text{s.t.} & \begin{bmatrix} \bar{\nu}(t)\\ \bar{y}(t) + \sigma(t) \end{bmatrix} = \begin{bmatrix}H_{L+n}(\nu^d) \\ H_{L+n}(y^d) \end{bmatrix}\alpha(t), \label{eq:ocp_dyn_of} \\
		&&& \begin{bmatrix} \bar{\nu}_{\left[-n,-1\right]}(t) \\ \bar{y}_{\left[-n,-1\right]}(t) \end{bmatrix} = \begin{bmatrix} \bar{\nu}^\ast_{\left[0,n-1\right]}(t-n) \\ y_{\left[t-n,t-1\right]} \end{bmatrix}, \label{eq:ocp_init_of} \\
		&&& \begin{bmatrix} \bar{\nu}_{\left[L-n,L-1\right]}(t) \\ \bar{y}_{\left[L-n,L-1\right]}(t) \end{bmatrix} = \begin{bmatrix} 0 \\ 0 \end{bmatrix}, \label{eq:ocp_term_of} \\
		&&&\norm{\tilde K}_\infty \eta_{A}^k\norm{\xi_t}_\infty + \norm{\tilde K}_\infty \sum_{i=0}^{k-1}\eta_{A}^{k-1-i}\eta_{B}\norm{\bar \nu_i (t)}_\infty + \norm{\bar \nu_k (t)}_\infty + \norm{\tilde K}_\infty \bar \delta_k \leq u_\mathrm{max}, \label{eq:ocp_ic_of} \\
		&&&\eta_{C}\eta_{A}^k\norm{\xi_t}_\infty + \eta_{C}\sum_{i=0}^{k-1}\eta_{A}^{k-1-i}\eta_{B}\norm{\bar \nu_i (t)}_\infty + \eta_{D}\norm{\bar \nu_k (t)}_\infty + \eta_{C}\bar \delta_k \leq y_\mathrm{max}, \label{eq:ocp_oc_of} \\
		&&&\forall k = 0,\dots,L-n-1.
	\end{alignat}
\end{subequations}
Problem \eqref{eq:ocp_of} is similar to \eqref{eq:ocp_sf} in the state-feedback case. The key difference is the new input and output constraint tightening in \eqref{eq:ocp_ic_of} and \eqref{eq:ocp_oc_of}. These constraints are now independent of $\alpha$ and $\sigma$. Instead, the tightening only involves the extended state at time $t$ and the input variables as well as the above defined constants. However, this comes at the price of potential conservatism, since \eqref{eq:ocp_ic_of} and \eqref{eq:ocp_oc_of} involve terms $\eta_A^k$, which is in general larger than terms of the form $\norm{\tilde{A}_K^k}_\infty$ that have been used in the constraint tightening of Section \ref{sec:state_feedback}. Setting up a constraint tightening similar to the one of \eqref{eq:ocp_sf} remains an interesting issue for future research. Moreover, the initial constraint \eqref{eq:ocp_init_of} and the terminal constraint \eqref{eq:ocp_term_of} hold over $n$ steps. This implies that the internal state of the underlying minimal realization corresponding to the prediction coincides with the initial state and the terminal state, respectively (compare the work of Markovsky and Rapisarda\cite{Markovsky2008}). Note that the constants in \eqref{eq:ocp_ic_of} and \eqref{eq:ocp_oc_of} can be computed purely from data following the approach that will be discussed in the next subsection.

We, again, close the loop by applying the optimal solution of \eqref{eq:ocp_of} in an $n$-step manner, i.e., $u_{t+k} = \tilde K\xi_{t+k} + \bar \nu_k^\ast (t)$ for $k=0,\dots,n-1$, where $\bar \nu_k^\ast(t)$ is the optimal solution of \eqref{eq:ocp_of} for the prediction step $k$. Note that $\xi_{t+k}$ contains the inputs $u_{[t+k-n,t+k-1]}$ and the measured outputs $y_{[t+k-n,t+k-1]}$. We are now ready to state practical exponential stability, and input as well as output constraint satisfaction of the closed loop. To this end, following the approach by Berberich et al.\cite{Berberich2020a}, we now consider the Lyapunov function
\begin{equation}
	V_t \coloneqq J_{L}^\ast(u_{\left[t-n,t-1\right]},y_{\left[t-n,t-1\right]}) + \gamma W(\xi_t),
\end{equation}
for some $\gamma > 0$, where $W(\xi)$ is an IOSS Lyapunov function, which exists due to detectability of $(A,C)$\cite{Cai2008}.
\begin{theorem}\label{th:of}
	Suppose that Assumption \ref{as:pe_of} holds. Then, for any $V_\mathrm{ROA} > 0$, there exist $\underline{\lambda}_\alpha$, $\overline{\lambda}_\alpha$, $\underline{\lambda}_\sigma$, $\overline{\lambda}_\sigma$ such that for all $\lambda_\alpha$, $\lambda_\sigma$ satisfying
	\begin{equation}
		\underline{\lambda}_\alpha \leq \lambda_\alpha \leq \overline{\lambda}_\alpha, \quad \underline{\lambda}_\sigma \leq \lambda_\sigma \leq \overline{\lambda}_\sigma,
	\end{equation}
	there exist $\bar{ \tilde w}_\mathrm{max}$, $\bar c_{pe} >0$ as well as a continuous, strictly increasing function $\beta : [0,\bar{ \tilde w}_\mathrm{max}] \to [0, V_\mathrm{ROA}]$ with $\beta(0) = 0$, such that for all $\tilde w_\mathrm{max}$ and $c_{pe}$ satisfying
	\begin{equation}
		\tilde w_\mathrm{max} \leq \min{\left\{\bar{ \tilde w}_\mathrm{max},\; \frac{\bar c_{pe}}{c_{pe}}\right\}}
	\end{equation}
	the following holds for the closed loop resulting from the application of the $n$-step MPC scheme:
	\begin{enumerate}
		\item[(i)] If $V_0 \leq V_\mathrm{ROA}$, then OCP \eqref{eq:ocp_of} is feasible at any time $t\geq 0$.
		\item[(ii)] For any initial condition satisfying $V_0 \leq V_\mathrm{ROA}$ it holds that $y_t\in\mathbb{Y}$ and $u_t\in\mathbb{U}$ for all $t\geq 0$, and $V_t$ converges exponentially to $V_t \leq \beta (\bar{ \tilde w}_\mathrm{max})$.
	\end{enumerate}
\end{theorem}
This result is similar to the state-feedback case (Theorem \ref{th:sf}). Also the proof works along the lines of the proof of Theorem \ref{th:sf}, where the candidate solution can be chosen analogously. The main difference lies in the modified constraint tightening. For a discussion on the role of the parameters in the above statement, we refer to the discussion below Theorem \ref{th:sf}.
\begin{proof}
	The proof is analogous to the proof of Theorem \ref{th:sf}. The only difference lies in the constraint tightening \eqref{eq:ocp_ic_of}, \eqref{eq:ocp_oc_of}. To show (i), we note that
	\begin{equation}
		\bar \nu'(t+n) = 
		\begin{bmatrix}
			\bar \nu^\ast_{[0,L-n-1]}(t) \\
			\bar \nu'_{[L-2n, L-n-1]}(t+n) \\
			0_n
		\end{bmatrix}
	\end{equation}
	and
	\begin{equation}
	\bar y'(t+n) = 
		\begin{bmatrix}
			y_{[t-n,t-1]} \\
			\hat y_{[t+n,t+L-n-1]}^\ast \\
			\bar y'_{[L-2n, L-n-1]}(t+n) \\
			0_n
		\end{bmatrix}
	\end{equation}
	are an input and output candidate solution. Here, $\hat y^\ast$ is the undisturbed output starting at $x_t$, resulting from an open-loop application of $\bar \nu^\ast(t)$; furthermore, $\bar \nu'_{[L-2n, L-n-1]}(t+n)$ is the input steering the system to the origin in $n$ steps and $\bar y'_{[L-2n, L-n-1]}(t+n)$ is the associated output. Note that, analogous to the proof of Theorem \ref{th:sf}, such an input exists due to controllability of the system. Moreover, we choose the candidate
	\begin{equation}
		\alpha'(t+n) = H_{u\hat x}^{y\dagger} \begin{bmatrix} \bar \nu'(t+n) \\ x_t\end{bmatrix},
	\end{equation}
	with 
	\begin{equation}	
		H_{u\hat x}^y = \begin{bmatrix} H_{L+n}\left(\nu^d\right) \\ H_1\left(\hat x^d_{[0,N-L-n]}\right) \end{bmatrix}.
	\end{equation}
	Furthermore, as a candidate for the slack variable, we choose
	\begin{equation}
		\sigma'(t+n) = H_{L+n}\left(y^d\right)\alpha'(t+n) - \bar y'(t+n).
	\end{equation}
	Thus, the candidate solution satisfies \eqref{eq:ocp_dyn_of}-\eqref{eq:ocp_term_of}. To show that \eqref{eq:ocp_ic_of} is also satisfied for the candidate solution, we note that $\bar \nu'_k(t+n) = \bar \nu^\ast_{k+n}(t)$ holds for $k=0,\dots,L-2n-1$. Thus, it holds for these $k$ that
	\begin{align}
		u_\mathrm{max} &\geq \norm{\tilde K}_\infty \eta_A^{k+n}\norm{\xi_t}_\infty + \norm{\tilde K}_\infty\sum_{i=0}^{k+n-1}\eta_A^{k+n-1-i}\eta_B\norm{\bar \nu_i^\ast (t)}_\infty + \norm{\bar \nu_{k+n}^\ast (t)}_\infty + \norm{\tilde K}_\infty\bar \delta_{k+n} , \\
		\begin{split}
		&\geq \norm{\tilde K}_\infty\eta_A^k\left( \eta_A^n\norm{\xi_t}_\infty + \sum_{i=0}^{n-1}\eta_A^{n-1-i}\eta_B\norm{\bar \nu_i^\ast (t)}_\infty + \bar \delta_{n} \right) + \norm{\tilde K}_\infty\sum_{i=0}^{k-1}\eta_A^{k-1-i}\eta_B\norm{\bar \nu_i' (t+n)}_\infty \\
		& \quad + \norm{\bar \nu_k' (t+n)}_\infty + \norm{\tilde K}_\infty\bar\delta_k, 
		\end{split} \\
		&\geq \norm{\tilde K}_\infty\eta_A^k\norm{\xi_{t+n}}_\infty + \norm{\tilde K}_\infty\sum_{i=0}^{k-1}\eta_A^{k-1-i}\eta_B\norm{\bar \nu_i' (t+n)}_\infty + \norm{\tilde K}_\infty\bar\delta_k + \norm{\bar \nu_k' (t+n)}_\infty,
	\end{align}
	where the last inequality holds due to $\xi_{t+n} = \tilde A_K^n + \sum_{i=0}^{n-1}\tilde A_K^{n-1-i}\tilde B\nu_i^\ast(t) + \sum_{i=0}^{n-1}\tilde A_K^{n-1-i}\tilde{w}_{t+i}$. Therefore, the candidate input satisfies \eqref{eq:ocp_ic_of} for $k=0,\dots,L-2n-1$. Showing that for these prediction steps \eqref{eq:ocp_oc_of} is satisfied by the candidate output works analogously. Moreover, showing that \eqref{eq:ocp_ic_of} and \eqref{eq:ocp_oc_of} are also satisfied for $k=L-2n,\dots,L-n-1$ can be done following analogous steps to the ones in the proof of Theorem \ref{th:sf}, by noting that due to the terminal condition and controllability there exists a sufficiently small $\tilde w_\mathrm{max}$ such that $\bar \nu'_{k}(t+n)$ and $\bar y'_{k}(t+n)$ become arbitrarily small, thus satisfying the tightened input and output constraint.
	
	To show (ii), we again follow the same arguments as in Theorem 3 by Berberich et al.\cite{Berberich2020a} to conclude practical exponential stability and recursive feasibility. It remains to show closed-loop input and output constraint satisfaction. Therefore, we note that
	\begin{align}
		u_\mathrm{max} & \geq \norm{\tilde K}_\infty \eta_A^{k}\norm{\xi_t}_\infty + \norm{\tilde K}_\infty\sum_{i=0}^{k-1}\eta_A^{k-1-i}\eta_B\norm{\bar \nu_i^\ast (t)}_\infty + \norm{\tilde K}_\infty\bar \delta_{k} + \norm{\bar \nu_{k}^\ast (t)}_\infty \\
		& \geq \norm{\tilde K\left(\tilde A_K^k\xi_t + \sum_{i=0}^{k-1} \tilde A_K^{k-1-i}\tilde B  \bar \nu_i^\ast (t) + \delta_{t+k} \right) + \bar \nu_k^\ast(t)}_\infty \\
		& = \norm{\tilde K\xi_{t+k} + \nu_{t+k}}_\infty
	\end{align}
	holds for $k=0,\dots,n-1$, which proves closed-loop input constraint satisfaction. Showing output constraint satisfaction works analogously.
\end{proof}

\subsection{Data-driven Estimation of System Constants}\label{ssec:system_constants_of}

In the following, we provide data-based procedures to compute (overapproximations of) the coefficients $\eta_A$, $\eta_B$, $\eta_C$, and $\eta_D$. First, we note that $\eta_A $ can be computed based on\cite{Wildhagen2021}, analogously to $\eta_{A,k}$ in Subsection \ref{ssec:system_constants}. The coefficients $\eta_B$, $\eta_C$, and $\eta_D$ can be computed similarly by modifying the approach from\cite{Wildhagen2021}.

Let us consider the data matrices
\begin{equation}
	\begin{split}
		X_+ &\coloneqq \begin{bmatrix} \xi^d_{n+1} & \xi^d_{n+2} & \dots & \xi^d_{N}\end{bmatrix}, \\
		X &\coloneqq \begin{bmatrix} \xi^d_{n} & \xi^d_{n+1} & \dots & \xi^d_{N-1}\end{bmatrix}, \\
		Y &\coloneqq \begin{bmatrix} y^d_{n} y^d_{n+1} & \dots & y^d_{N-1}\end{bmatrix}, \\
		U &\coloneqq \begin{bmatrix} \nu^d_{n} \nu^d_{n+1} & \dots & \nu^d_{N-1}\end{bmatrix},
	\end{split}
\end{equation}
where $\xi^d_k = \begin{bmatrix} u_{k-n}^{d\top} & \dots & u_{k-1}^{d\top} & y_{k-n}^{d\top} & \dots & y_{k-1}^{d\top}\end{bmatrix}^\top$ for $k=n,\dots,N$. We write $F\succeq 0$ if $F$ is a symmetric and positive semi-definite matrix. In order to compute overapproximations of the respective system constants, we solve the optimization problem
\begin{subequations}
\label{eq:estimation_eta}
	\begin{alignat}{4}
	&&\min_{\tau,\bar \sigma^2}&\quad \bar\sigma^2, \\
	&&\text{s.t.} &\quad P_1(\bar\sigma^2) - \tau P_2  \succeq 0, \\
	&&& \quad \bar\sigma^2\geq 0, \\
	&&& \quad \tau \geq 0,
	\end{alignat}
\end{subequations}
where $P_1$ and $P_2$ are placeholders, which have to be defined for the different coefficients as follows
\begin{align}
	\eta_A: \quad P_1(\bar\sigma^2) & = \begin{bmatrix} -I & 0 & 0 \\ 0 & 0 & 0 \\ 0 & 0 & \bar\sigma^2I \end{bmatrix}, &P_2 &= \begin{bmatrix} -\begin{bmatrix} X \\ U \end{bmatrix}^\top & X_+^\top \\ 0 & I \end{bmatrix}^\top \cdot\begin{bmatrix} -\tilde{E}\tilde{E}^\top & 0 \\ 0 & n\tilde w_\mathrm{max}^2N \end{bmatrix} \cdot \begin{bmatrix} -\begin{bmatrix} X \\ U \end{bmatrix}^\top & X_+^\top \\ 0 & I \end{bmatrix}, \\
	\eta_B: \quad P_1(\bar\sigma^2) & = \begin{bmatrix} 0 & 0 & 0 \\ 0 & -I & 0 \\ 0 & 0 & \bar\sigma^2I \end{bmatrix}, &P_2 & = \begin{bmatrix} -\begin{bmatrix} X \\ U \end{bmatrix}^\top & X_+^\top \\ 0 & I \end{bmatrix}^\top \cdot\begin{bmatrix} -\tilde{E}\tilde{E}^\top & 0 \\ 0 & n\tilde w_\mathrm{max}^2N \end{bmatrix} \cdot \begin{bmatrix} -\begin{bmatrix} X \\ U \end{bmatrix}^\top & X_+^\top \\ 0 & I \end{bmatrix}, \\
	\eta_C: \quad P_1(\bar\sigma^2) & = \begin{bmatrix} -I & 0 & 0 \\ 0 & 0 & 0 \\ 0 & 0 & \bar\sigma^2I \end{bmatrix}, &P_2 &= \begin{bmatrix} -\begin{bmatrix} X \\ U \end{bmatrix}^\top & Y^\top \\ 0 & I \end{bmatrix}^\top \cdot \begin{bmatrix} -I^\top & 0 \\ 0 & n\tilde w_\mathrm{max}^2N \end{bmatrix} \cdot \begin{bmatrix} -\begin{bmatrix} X \\ U \end{bmatrix}^\top & Y^\top \\ 0 & I \end{bmatrix}, \\
	\eta_D: \quad P_1(\bar\sigma^2) & = \begin{bmatrix} 0 & 0 & 0 \\ 0 & -I & 0 \\ 0 & 0 & \bar\sigma^2I \end{bmatrix}, &P_2 &= \begin{bmatrix} -\begin{bmatrix} X \\ U \end{bmatrix}^\top & Y^\top \\ 0 & I \end{bmatrix}^\top \cdot\begin{bmatrix} -I^\top & 0 \\ 0 & n\tilde w_\mathrm{max}^2N \end{bmatrix}\cdot \begin{bmatrix} -\begin{bmatrix} X \\ U \end{bmatrix}^\top & Y^\top \\ 0 & I \end{bmatrix}.
\end{align}

We denote the solutions of \eqref{eq:estimation_eta} as $\bar \sigma_A^2$, $\bar \sigma_B^2$, $\bar \sigma_C^2$, $\bar \sigma_D^2$ for the respective configuration of $P_1$ and $P_2$. It is straight forward to show that we obtain $\eta_A$, $\eta_B$, $\eta_C$, and $\eta_D$ via $\eta \leq \sqrt{n}\bar\sigma$. For a detailed discussion of the approach we refer to\cite{Wildhagen2021}.

\section{Numerical example}\label{sec:examples}

\begin{figure}
	\centering
	\begin{minipage}{.49\textwidth}
			\includegraphics[width=1\linewidth]{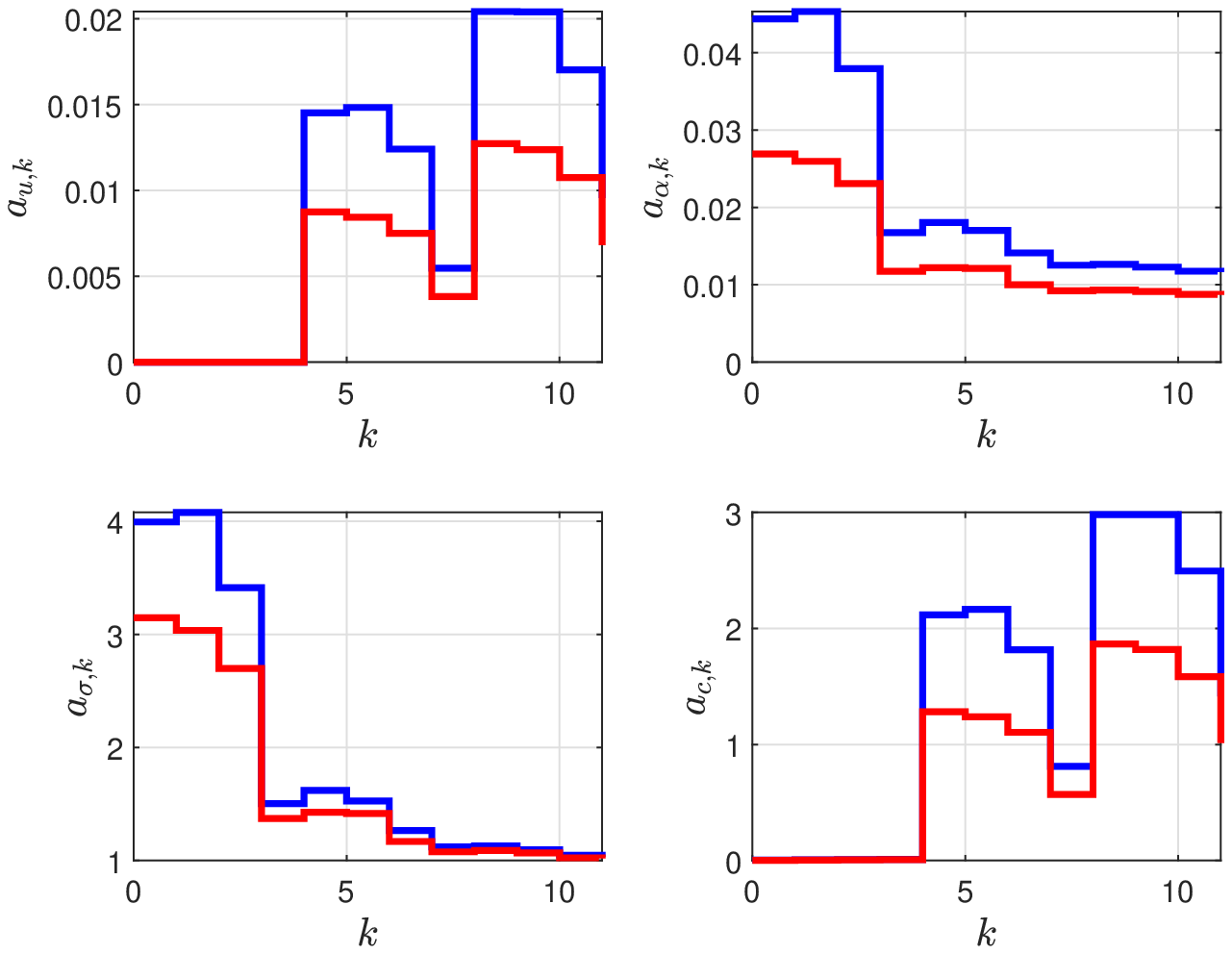}
	\end{minipage}
	\begin{minipage}{.49\textwidth}
	\centering
		\includegraphics[width=1\linewidth]{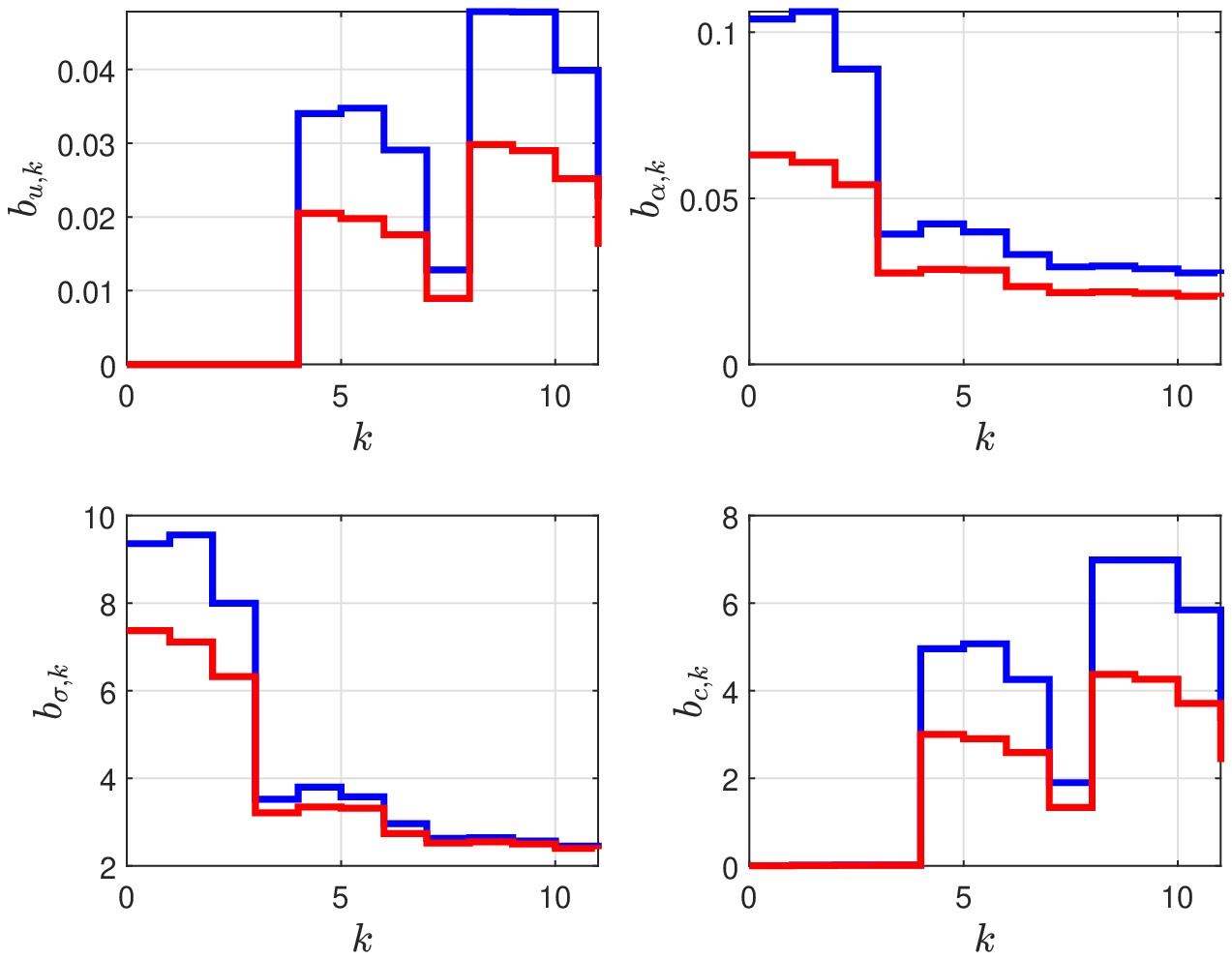}
	\end{minipage}
	\caption{Coefficients for the state and input constraint tightening \eqref{eq:ocp_sc}, \eqref{eq:ocp_ic}. The red lines correspond to the ideal coefficients that can be computed if perfect model knowledge is available. The blue lines correspond to the coefficients computed purely from data.}
	\label{fig:param}
\end{figure}

\begin{figure}
	\centering
		\includegraphics[width=.5\linewidth]{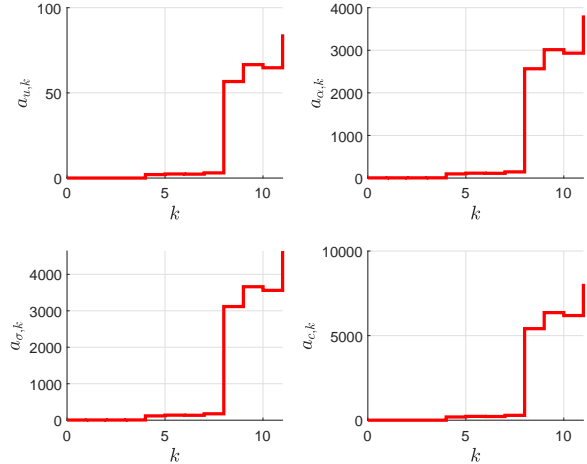}
	\caption{Ideal coefficients for the state constraint tightening \eqref{eq:ocp_sc} without pre-stabilizing controller}
	\label{fig:param_a_unstable}
\end{figure}

\begin{figure}
	\centering
		\includegraphics{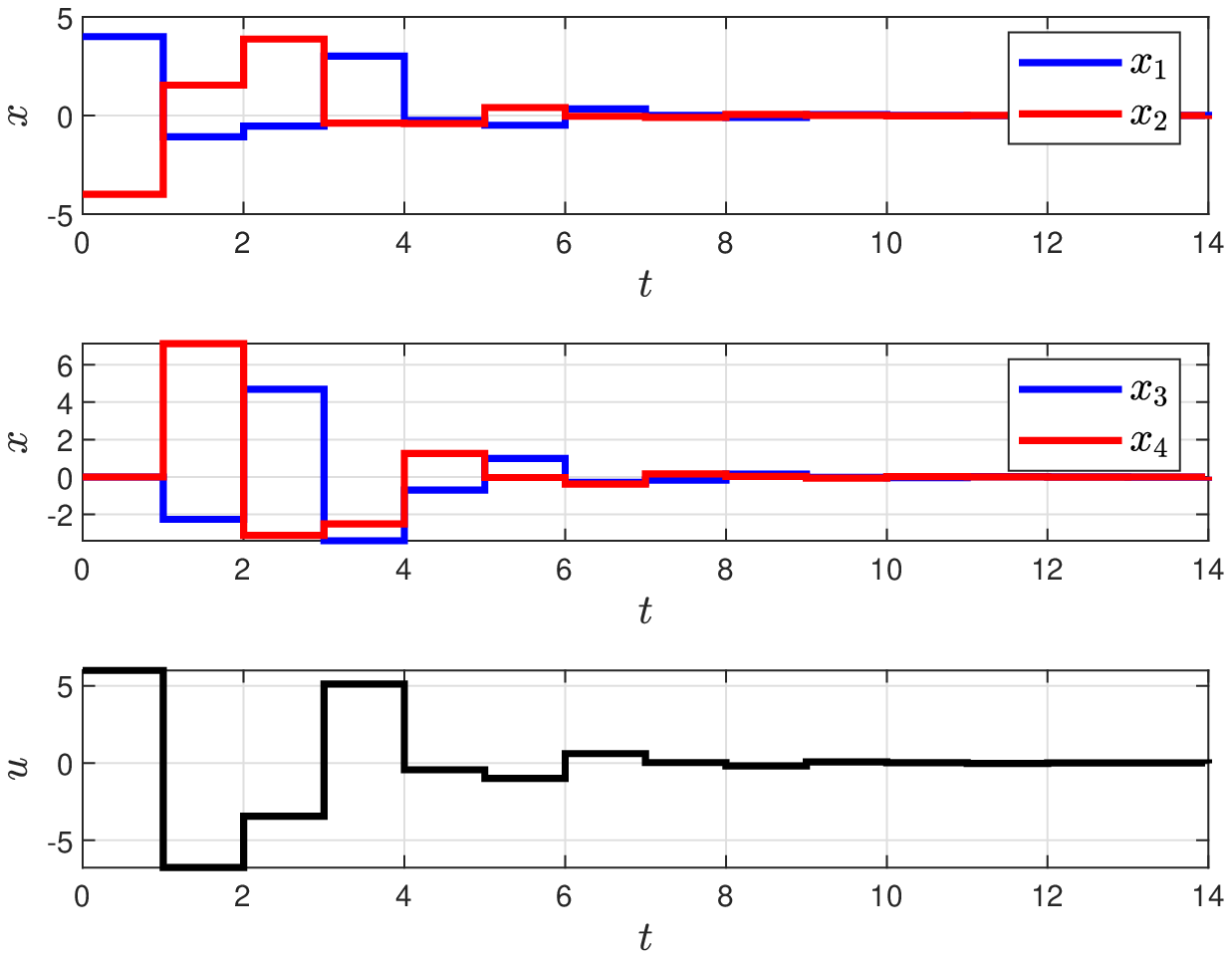}
	\caption{Simulation}
	\label{fig:sim}
\end{figure}

As an example, we consider the two mass-spring-system suggested by Wie et al.\cite{Wie1992}, with the masses $m_1 = 0.5\, \mathrm{kg}$, $m_2 = 1\, \mathrm{kg}$ and the spring constant $k = 2\, \frac{\mathrm{kg}}{\mathrm{s}^2}$. Discretizing the system with a sampling time of $1\, \mathrm{s}$ yields the matrices
\begin{equation}
	A = 
	\begin{bmatrix}
	 -0.1799 & 1.1799 & 0.507 & 0.493 \\
    0.59 & 0.41 & 0.2465 & 0.7535 \\
   -1.0421 & 1.0421 & -0.1799 & 1.1799 \\
    0.5211 & -0.5211 & 0.59 & 0.41 \\
	\end{bmatrix}, \quad B = 
	\begin{bmatrix} 
		0.7266 \\ 0.1367 \\ 1.014 \\ 0.493
	\end{bmatrix}.
\end{equation}
We assume that full state measurements are available and that $w_\mathrm{max} = 10^{-3}$, $u_\mathrm{max} = 10$, $x_\mathrm{max} = 10$ hold for the constraint sets, where the process disturbance $w_k$ acting on the system, during the data generation and in closed-loop operation, at time $k$ is sampled uniformly from $\mathbb{W}$. As the matrix $A$ has two eigenvalues on the unit disc, the usage of a pre-stabilizing controller is expected to be be advantageous in order to set up the MPC scheme introduced in Section \ref{sec:state_feedback}. First, we collect apply a PE input sequence of length $50$ to the open-loop system and measure the corresponding state sequence. Thereafter, we employ Theorem 1 by Berberich et al.\cite{Berberich2020d} to compute a robust linear quadratic regulator for the system based on the available noisy data (using diagonal multipliers to describe the disturbance bound, compare Equation (21) in\cite{Berberich2020d}). This yields the state-feedback gain 
\begin{equation*}
	K = \begin{bmatrix} 0.4345 & -0.8439 & -0.3665 & -0.6986 \end{bmatrix}
\end{equation*}
which serves as the pre-stabilizing controller, leading to all eigenvalues of $A_K = A + BK$ being located strictly inside the unit disc. 
 
The goal is to set up the OCP \eqref{eq:ocp_sf} with prediction horizon $L=12$, a total amount of $N=50$ data points (of the pre-stabilized system) in the Hankel matrices, and the parametrization $Q=I_n$, $R=1$, $\lambda_\alpha = \lambda_\sigma = 100$. To this end, we have to compute the constants \eqref{eq:tight_constants_1} and \eqref{eq:tight_constants_2} from data. In order to do so, we apply an input sequence of length $N'=5000$, which is uniformly sampled from $\mathbb{U}$, to the pre-stabilized system and collect the corresponding $N'$ state measurements. Using these data sequences, we compute $\rho_{A,k}$ for $k=0,\dots, L-1$ and $\bar d_k$ for $k=0,\dots,N-1$ following the approach mentioned in Subsection \ref{ssec:system_constants}. Note that the estimation of the system constants also works for a smaller amount of data, at the price of a more conservative overapproximation of these constants. However, for a good approximation of these constants, we need a much larger amount of data than for the prediction via the Hankel matrices (i.e., $N'\gg N$). We now choose an input-state-sequence of length $N$ from the collected data (of total length $N'$), for which the input sequence is persistently exciting of order $L+n+1$, to construct the Hankel matrices and approximate the constant $c_{pe}$ as described in Subsection \ref{ssec:system_constants}. Moreover, we apply the method from Subsection \ref{ssec:system_constants}, with $\lambda_\alpha'=\lambda_\sigma' = 1$, in order to compute an approximation of the controllability constant $\Gamma$. With these approximations of the system constants, we compute the coefficients in \eqref{eq:tight_constants_1} and \eqref{eq:tight_constants_2}. The resulting coefficients as well as the ideal coefficients that would be computable if perfect model knowledge was available, can be found in Figure \ref{fig:param}, where the red lines correspond to the ideal coefficients, and the blue lines to the coefficients computed from data. It can be seen that the coefficients computed from data yield over approximations of the "real" coefficients. This is mainly due to the fact that the procedure in Subsection \ref{ssec:system_constants} only yields overapproximations of the constants $\rho_{A,k}$ and $\bar d_k$. The approximation of $c_{pe}$ by its disturbed counterpart is accurate for the present example,compare the discussion in Subsection \ref{ssec:system_constants}.

Considering the parameter $b_{c,k}$, it can be seen that for larger $k$, this coefficient already is close to $u_\mathrm{max}$. Even though there is a non-monotonicity in $k$, which results from the usage of an $n$-step MPC scheme and the corresponding recursive definition of the constants in \eqref{eq:tight_constants_1} and \eqref{eq:tight_constants_2}, it is clearly visible that $b_{c,k}$ tends to increase for larger $k$. Thus, for larger prediction horizons or larger disturbance bounds $w_\mathrm{max}$, this parameter would lead to $b_{c,k} > u_\mathrm{max}$, which would render \eqref{eq:ocp_sf} infeasible due to \eqref{eq:ocp_ic}. The reason for this conservatism in the constraint tightening lies in the fact that submultiplicativity and the triangle inequality were exploited multiple times in the proof of recursive feasibility and constraint satisfaction. Moreover, the estimates for $\rho_{A,k}$ and $\bar d_k$ only yield overapproximations of the parameters and for the sake of recursive feasibility, $x_\mathrm{max}$ instead of $\norm{x_t}_\infty$ is used to define the coefficients \eqref{eq:tight_constants_2}. 

As a motivation for the usage of a pre-stabilizing controller, the coefficients for $K=0$ (i.e., without input constraint tightening) are plotted in Figure \ref{fig:param_a_unstable}. Note that even for $k = 4$ it holds that $a_{c,k} \approx 233$ which already exceeds $x_\mathrm{max}$ and would thus lead to an infeasible OCP even for the smallest possible prediction horizon of the $n$-step scheme, $L=4$.

The states and inputs resulting from the $n$-step scheme in closed loop starting at $x_0 = \begin{bmatrix} 4 & -4 & 0 & 0 \end{bmatrix}^\top$ can be found in Figure \ref{fig:sim}. It can be seen that the predictive control scheme works as desired, meaning it stabilizes the states at the origin, while satisfying the state and input constraints.

\section{Conclusion}\label{sec:conclusion}

In this paper, we introduced a data-driven predictive control scheme relying on predictions based on a priori measured data, structured in Hankel matrices. The scheme is capable of stabilizing the origin of an LTI system, even in the presence of process disturbances acting on the system state. To this end, a constraint tightening is proposed, which can be set up using only a priori measured data. The presented MPC scheme allows for the usage of a pre-stabilizing controller and an associated input constraint tightening, which enables the use also for a priori unstable systems. Closed-loop recursive feasibility, practical exponential stability, and constraint satisfaction of the control scheme is shown. Moreover, the MPC, initially introduced for available state measurements, is extended to the case that only output measurements are available. The numerical experiments illustrated the applicability of the proposed approach and underlined the necessity to include a pre-stabilization and corresponding input constraint tightening in order to design a feasible controller. Interesting issues for future research include the development of less conservative constraint tightenings, in particular in the output-feedback case, as well as a data-based techniques for obtaining estimates of $c_{pe}$ and the controllability constant $\Gamma$ from noisy data, which are guaranteed overapproximations of the real system constants.


\bibliography{references}
\bibliographystyle{unsrt}

\end{document}